\documentclass[runningheads]{llncs}
\makeatletter
\newcommand{\@chapapp}{\relax}%
\makeatother
\usepackage[title]{appendix}

\usepackage{lipsum}
\usepackage{float}
\restylefloat{figure}
\usepackage{textcomp}
\usepackage{xcolor}

\usepackage{graphicx}
\usepackage{amsmath,amssymb,amsfonts}

\usepackage{amsthm}
\usepackage{nicefrac}
\usepackage{cases}
\usepackage{algorithm,algpseudocode,float}
\algdef{SE}[IF]{NoThenIf}{EndIf}[1]{\algorithmicif\ #1}{\algorithmicend\ \algorithmicif}
\usepackage{tikz}
\usetikzlibrary{calc}
\usetikzlibrary{arrows.meta, decorations.pathreplacing}

\newtheorem{thm}{Theorem}
\newtheorem{obs}[thm]{Observation}
\theoremstyle{definition}
\newtheorem{defn}[thm]{Definition}
\newtheorem{prop}[thm]{Proposition}

\DeclareMathOperator{\sign}{Sign}
\DeclareMathOperator{\diam}{diam}
\newcommand\ceil[1]{\lceil#1\rceil}
\newcommand\floor[1]{\lfloor#1\rfloor}
\newcommand\card[1]{\lvert#1\rvert}
\newcommand{\etal}{\textit{et al}. }
\newcommand{\ie}{\text{i.e}., }
\newcommand{\orient}[1]{\overrightarrow{#1}}
\def\SN{\ensuremath{S_n}}
\def\USN{\ensuremath{\orient{S_n}}}

\begin{document}
\title{Oriented Diameter of Star Graphs}
\author{K. S. Ajish Kumar\inst{1}\and
Deepak Rajendraprasad\inst{2} \and
K. S. Sudeep\inst{3}}
\authorrunning{K. S. Ajish Kumar, Deepak Rajendraprasad, K. S. Sudeep}

\institute{Department of Electronics and Communication Engineering, National Institute of Technology Calicut, India \and
Department of Computer Science and Engineering, Indian Institute of Technology Palakkad, India
\and
Department of Computer Science and Engineering, National Institute of Technology Calicut, India}
\maketitle

\begin{abstract}

An {\em orientation} of an undirected graph $G$ is an assignment of exactly
one direction to each edge of $G$. Converting two-way traffic networks to
one-way traffic networks and bidirectional communication networks to 
unidirectional communication networks are practical instances of graph
orientations. In these contexts minimising the diameter of the resulting
oriented graph is of prime interest. 

The $n$-star network topology was proposed as an alternative to the hypercube
network topology for multiprocessor systems by Akers and Krishnamurthy [IEEE
Trans. on Computers (1989)].  The {\em $n$-star graph} $\SN$ consists of $n!$
vertices, each labelled with a distinct permutation of $[n]$. Two vertices are
adjacent if their labels differ exactly in the first and one other position.
$\SN$ is an $(n-1)$-regular, vertex-transitive graph with diameter
$\floor{3(n-1)/2}$.  Orientations of $\SN$, called {\em unidirectional star
graphs} and distributed routing protocols over them were studied by Day and
Tripathi [Information Processing Letters (1993)] and Fujita [The First
International Symposium on Computing and Networking (CANDAR 2013)]. Fujita
showed that the (directed) diameter of this unidirectional star graph $\USN$ is
at most $\ceil{5n/2} + 2$.

In this paper, we propose a new distributed routing algorithm for the same
$\USN$ analysed by Fujita, which routes a packet from any node $s$ to any node
$t$ at an undirected distance $d$ from $s$ using at most $\min\{4d+4, 2n+4\}$
hops. This shows that the (directed) diameter of $\USN$ is at most $2n+4$.  We
also show that the diameter of $\USN$ is at least $2n$ when $n \geq 7$, thereby showing that our upper bound is tight up to an additive
factor.

\end{abstract}

\section{Introduction}

Let $G=(V, E)$ be an undirected graph with vertex set $V$ and edge set $E$.
An orientation $\orient{G}$ of $G$ is a directed graph obtained by assigning
exactly one direction to each edge of $G$. An orientation is called a
\emph{strong orientation} if the resulting directed graph is strongly
connected. A directed graph $\orient{G}$ is said to be strongly connected, if there exists at least one directed path from every vertex of $\orient{G}$ to every other vertex. There can be many strong orientations for $G$. 
The smallest diameter among all possible strong orientations of $G$ is called the \emph{oriented diameter} of $G$, denoted by $\orient{diam}(G)$. 
That is, $\orient{diam}(G) = \min \{
\diam(\orient{G}) \vert ~ \orient{G} \text{ is a strong orientation of G}\}$. 

The research on strong orientations dates back to 1939 with Robbins \cite{robbins1939theorem}, solving the {\em One Way Street} problem. Given the road network of city, the One Way Street problem poses the following question:  
Is it possible to implement one way traffic in every street without compromising the accessibility of any of the junctions of the network? 
Robbins proved that the necessary and sufficient condition for the existence of
a strong orientation of a graph $G$ is the $2$-edge connectivity of $G$. A
$2$-edge connected graph is one that cannot be disconnected by
removal of a single edge. The research on orientations that minimise the
resulting distances was initiated by Chv\'{a}tal and Thomassen in 1978
\cite{chvatal1978distances}. They proved that, for every undirected graph $G$
there exists an orientation $\orient{G}$ such that for every edge $(u,v)$ which
belongs to a cycle of length $k$, either $(u,v)$ or $(v,u)$ belong to a cycle
of length $h(k)$ in $\orient{G}$, where $h(k) = (k-2)2^{\floor{\frac{(k-1)}{2}}}+2$.
They also showed that every $2$-edge connected undirected graph of diameter $d$
will possess an orientation with diameter at most $2d^2+2d$.  Further, they
proved that it is NP-hard to decide whether an undirected graph possesses an
orientation with diameter at most $2$.

Fomin et al. \cite{fomin2004complexity} continues the algorithmic study on oriented diameter on chordal graphs.  They show that every chordal graph $G$ has an oriented diameter at most $2 \diam(G)+1$. This result proves that the oriented diameter problem is $(2,1)$-approximable for chordal graphs.
A polynomial time algorithm for finding the oriented diameter of planar graphs was given by Eggemann \cite{eggemann2009minimizing}. 
Fomin \etal \cite{fomin2004free} have proved that the oriented diameter of every AT-free bridgeless connected graph $G$ is at most $2 \diam(G)+11$ and for every interval graph $G$, it is at most $\frac{5}{4}\diam(G)+\frac{29}{2}$.
Dankelmann et al., \cite{dankelmann2018oriented} proved that every $n$-vertex bridgeless graph with maximum degree $\Delta$ has oriented diameter at most $n-\Delta+3$. For balanced bipartite graphs (a bipartite graph with equal number of vertices on both halves of the bipartition), they prove a better bound of $n-2\Delta+7$.
The problems of finding strong orientations that minimize the parameters such as diameter, distance between pairs of vertices \emph{etc.}, have been investigated for other restricted subclasses of 
graphs like $n$-dimensional hypercube \cite{chou1990uni}, torus \cite{konig1998diameter}, star graph \cite{day1993unidirectional,fujita2013oriented}, and (n,k)-star graph \cite{cheng2002unidirectional}.

Oriented diameter problem finds a significant application in parallel computing.
In interconnection networks of parallel processing systems, the processing elements are connected together using fibre optic links that support high bandwidth, high speed and long distance data communication. However, the
optical transmission medium suffers from the drawback that the links are
inherently unidirectional \cite{chou1990uni}.  In the case of optical links, a
naive strategy to achieve bidirectional communication is to use two separate
optical links between every pair of communication entities. But, such a naive
approach increases the hardware complexity and cost of the network. On the
other hand,  unidirectional communication links are simple and cost effective
but require more number of intermediate communication hops to establish
bidirectional communication. Thus, the average interprocessor communication
delay is generally more in the case of unidirectional interconnection networks.
However, unidirectional interconnection networks might be the best
choice if we can trade off communication delay with cost and hardware
complexity of the network. 

\subsection{The $n$-star graph ($S_n$)}

In \cite{akers1989group}, Akers and Murthy presented a group theoretic model called {\em Cayley Graph Model} for designing symmetric interconnection networks. In parallel computing the interconnection networks provide an efficient communication mechanism among the processors and the associated memory. For a finite group $\Gamma$ and a set $S$ of generators of $\Gamma$, the Cayley Graph $D = D(\Gamma ,S)$ is the directed graph defined as follows. The vertex-set of $D$ is $\Gamma$. There is an arc from a vertex $u$ to a vertex $v$ in $D$, if and only if there exist a generator $g$ in $S$ such that $ug=v$. Further, if the inverse of every element in $S$ is also in $S$, the two directed edges between $u$ and $v$ are replaced by a single undirected edge, resulting in an undirected graph. In \cite{akers1987star}, Akers and Murthy proposed a new
symmetric graph, called {\em Star Graph}, $\SN$. Let $G$ be a group with elements being all permutations of
the set $\{1, 2, \ldots ,n\}$ and group operation being composition. The star
graph $\SN$ is a Cayley graph on $G$ with generator set
$S=\{g_2, g_3,\ldots, g_n\}$, where $g_i$ is the permutation obtained by 
swapping the first and
$i^{th}$ value of the identity permutation. It is easy to see that, $\SN$ has
degree $n-1$, and it has been shown that the diameter of $\SN$ is
$\floor{3(n-1)/2}$ \cite{akers1987star}. The star graph has many desirable properties of a good interconnection network such as symmetry (vertex transitivity), small diameter, small degree and large connectivity. A symmetric interconnection network allows the use of same routing algorithm for every node, while a small degree reduces the cost of the network. Further, a small diameter reduces overall communication delay and large connectivity offers good fault tolerance.

Two different strong orientation schemes have been proposed for $\SN$. The first one was by K. Day and A. Tripathi
\cite{day1993unidirectional}.  They showed that the diameter of their
orientation is at most $5(n-2)+1$. The second orientation scheme was proposed
by S. Fujita \cite{fujita2013oriented}. The diameter of this orientation scheme
was shown to be at most $\ceil{5n/2}+2$. We observe that these two orientation
schemes are essentially the same. Both the schemes partition the set of
generators into nearly equal halves.  The edges due to first set of generators are oriented
from the odd permutation to the even permutation and those due to the second set of
generators are oriented in the opposite direction. The difference between the
two orientation schemes lies in the way by which the two schemes partition the
set of generators. The Day-Tripathi scheme splits the set of generators based
on the parity of $i$ of a generator $g_i$, \ie generators with odd parity for
$i$ belong to the first set and even parity for $i$ belong to the second set.
In the case of Fujita's orientation, the first partition consists of generators
from $g_2$ to $g_k, k = \ceil{(n-1)/2} + 1$, whereas, the second partition
consists of generators from $g_{k+1}$ to $g_n$. The details of the two
orientation schemes described above are depicted in
Fig.~\ref{Day-Tripathi-Fujita}, for two nodes with labels 12345 and 21345, and
their neighbours in $S_5$.

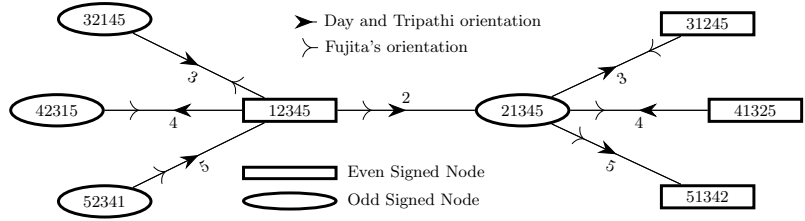
\begin{figure*}[htp]
\centering
\scalebox{.7}{\begin{tikzpicture}
\usetikzlibrary{arrows}
\usetikzlibrary{decorations.markings,arrows.meta}
\usetikzlibrary{shapes}
\tikzset{middlearrow/.style={
        decoration={markings,
        	mark= at position .25 with {\arrow{#1}} ,
        },
        postaction={decorate}
    }
}
\tikzset{endarrow/.style={
        decoration={markings,
        	mark= at position .5 with {\arrow{#1}} ,
        },
        postaction={decorate}
    }
}
\tikzstyle{every node}=[ultra thick, draw=black, minimum width=50pt,
    align=center]
    \tikzset{
   ultra thin/.style= {ellipse,line width=1.6pt},
    ultra thick/.style={rectangle,line width=1.6pt}
}

\node[ultra thick] (a) {12345};
\node[ultra thin,left=100pt, below=40pt] (b) at (a) {52341};
\node[ultra thin,left=100pt] (c) at (a) {42315};
\node[ultra thin,left=100pt, above=40pt] (d) at (a) {32145};
\node[ultra thin,right=100pt] (e) at (a) {21345};
\node[ultra thick,right=100pt, below=40pt] (f) at (e) {51342};
\node[ultra thick,right=100pt] (g) at (e) {41325};
\node[ultra thick,right=100pt, above=40pt] (h) at (e) {31245};

\draw[middlearrow={>[scale=2.0]}] (b) -- node[sloped,font=\small,draw=none,below]{5}(a);
\draw[endarrow={Stealth[scale=2.0]}] (b) -- (a);

\draw[middlearrow={>[scale=2.0]}] (c) -- node[sloped,font=\small,draw=none,below]{4}(a);
\draw[endarrow={Stealth[scale=2.0]}] (a) -- (c);

\draw[middlearrow={>[scale=2.0]}] (a) -- node[sloped,font=\small,draw=none,below]{3}(d);
\draw[endarrow={Stealth[scale=2.0]}] (d) -- (a);

\draw[middlearrow={>[scale=2.0]}] (a) -- node[sloped,font=\small,draw=none,above]{2}(e);
\draw[endarrow={Stealth[scale=2.0]}] (a) -- (e);

\draw[middlearrow={>[scale=2.0]}] (e) -- node[sloped,font=\small,draw=none,below]{5}(f);
\draw[endarrow={Stealth[scale=2.0]}] (e) -- (f);

\draw[middlearrow={>[scale=2.0]}] (e) -- node[sloped,font=\small,draw=none,below]{4}(g);
\draw[endarrow={Stealth[scale=2.0]}] (g) -- (e);

\draw[middlearrow={>[scale=2.0]}] (h) -- node[sloped,font=\small,draw=none,below]{3}(e);
\draw[endarrow={Stealth[scale=2.0]}] (e) -- (h);
\node[draw=none, minimum width=2pt, above=30pt] (x) at (a) {};
\node[draw=none, minimum width=2pt, right=1mm, label=right:{\,\,\,Fujita's orientation}] (w) at (x) {};
\draw[middlearrow={>[scale=2.0]}] (w) -- ++(.30cm,0)(x);

\node[draw=none, minimum width=2pt, above=10pt] (y) at (x) {};
\node[draw=none, minimum width=2pt, right=.8mm, label=right:{\,\,\,Day and Tripathi orientation}] (z) at (y) {};
\draw[endarrow={Stealth[scale=2.0]}] (z) -- ++(.30cm,0) (y);

\node[ultra thick, below=30pt, label=right:{\,Even Signed Node}] (p) at (a){};
\node[ultra thin, below=10pt, label=right:{\,Odd Signed Node}] (q) at (p) {}; 

\end{tikzpicture}}
\caption{Day-Tripathi and Fujita orientation schemes for $S_5$}
\label{Day-Tripathi-Fujita}
\end{figure*}

In this paper, we propose a new distributed routing algorithm for the same
$\USN$ analysed by Fujita. We show that the proposed algorithm routes a 
packet in $\USN$ from any node $s$ to any other node $t$ using at most 
$\min\{4d +4, 2n+4\}$ hops, where $d$ is the distance between $s$ and 
$t$ in $\SN$.  In particular, this shows that the (directed) diameter of 
$\USN$ is at most $2n+4$, which is an improvement over Fujita's upper bound. We also show that the diameter of $\USN$ is at least $2n$ when $n \geq 7$, thereby showing that our upper bound is tight up to an additive factor. We do not believe that either of the above orientations of $\SN$ are optimal in terms of achieving the minimum (directed) diameter. In fact, we believe that the oriented diameter of $\SN$ is $3n/2 + O(1)$. 

\section{Preliminaries}
\subsection{Graph Terminology}
Some of the basic definitions in graph theory which are required to understand
the details of this work are explained in this section.  Let $G=(V,E)$ be any
undirected graph with vertex-set $V$ and edge-set $E$.  Two vertices of
$G$ are called {\em neighbours} when they are connected by an edge. The {\em
degree} of a vertex $u$ is the number of neighbours of $u$. If all the vertices
of $G$ have the same degree, $G$ is called {\em regular}. The {\em distance} between two nodes $u$ and $v$, denoted by $d(u,v)$, is the number of edges along a shortest path between $u$ and $v$. The {\em diameter} of $G$, denoted by $\diam({G})$, is the maximum of $d(u,v)$ among all $u,v \in V$. An {\em automorphism} of $G$ is a permutation $\pi$ of $V$ such that  for every pair of vertices $u,v \in V$, $\{u,v\}$ is an edge in $E$, if and only if $\{\pi(u),\pi(v)\}$ is an edge in $E$. Two vertices $u$ and $v$ of $G$ are said to be \emph{similar} if there is an automorphism $\pi$ of $G$ with $\pi(u) = v$. $G$ is {\em vertex-transitive} when every pair of vertices in $G$ are similar. 

Let $D=(V,E)$ be any directed graph with vertex-set $V$ and edge-set $E$. If
$(u,v)$ is an edge (arc) in $E$, then $u$ is called an {\em in-neighbour} of
$v$ and $v$ is called an {\em out-neighbour} of $u$. The {\em in-degree} and
{\em out-degree} of a vertex $u$ are, respectively, the number of in-neighbours
and out-neighbours of $u$. The {\em distance} from a node $u$ to a node $v$,
denoted by $\orient{d}(u,v)$, is the number of edges along a shortest directed
path from $u$ to $v$. The {\em diameter} of $D$ ($\diam({D})$), is the maximum of
$\orient{d}(u,v)$ among all $u,v \in V$. Automorphism and Vertex-transitivity among directed graphs are defined similar to that of undirected graphs. 

\subsection{Cycle Structure of Permutations}
Let $\pi$ be a permutation of $\{1,\ldots,n\}$. The {\em sign} of $\pi$,
denoted by $\sign(\pi)$, is defined as the parity of the number of inversions
in $\pi$, that is $x,y \in \{1,\ldots,n\}$, such that $x < y$ and $\pi(x) >
\pi(y)$. A {\em cycle} $(a_0, \ldots, a_{k-1})$ is a permutation $\pi$ of
$\{a_0, \ldots, a_{k-1}\}$ such that $\pi(a_i) = a_{i+1}$ where addition is
modulo $k$.  Two cycles are disjoint if they do not have common elements. Every
permutation of $[n]$ has a unique decomposition into a product of disjoint
cycles. The sign of a permutation turns out to be the parity of the number of
even-length cycles in that permutation. 

One hop in an $n$-star graph corresponds to moving from a permutation $\sigma$
to another permutation $\pi$, by exchanging the value $\sigma(1)$ with a value
$\sigma(k)$, $k \in \{2, \ldots n\}$. We would like to make some observations
about the cycle structure of $\pi$ and $\sigma$. In the case when $1$ and $k$
belong to the same cycle of $\sigma$, this cycle gets broken into two disjoint
cycles in $\pi$ (Fig.~\ref{cycle-break}).
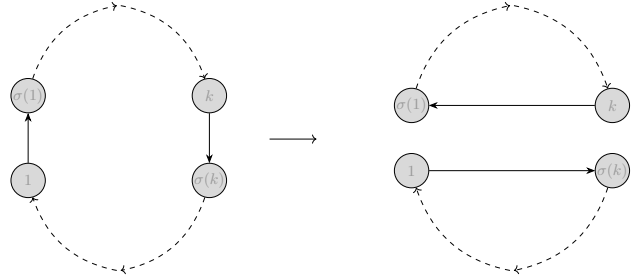
\begin{figure}[htp]
\centering
\scalebox{.65}{\begin{tikzpicture}[
      mycircle/.style={
         circle,
         draw=black,
         fill=gray,
         fill opacity = 0.3,
         text opacity=1,
         inner sep=0pt,
         minimum size=20pt,
         font=\small},
         state/.style={circle,inner sep=0pt, minimum size=2pt},
      myarrow/.style={-Stealth},
      node distance=2cm and 2cm
      ]
     
  \node[state](c1){};
  \node[state,below =5.4cm] (c4) at (c1){};   
  \node[mycircle, below right=1.6cm] (c2) at (c1) {$k$};
\node[mycircle,above right = 1.6cm] (c3) at (c4) {$\sigma(k)$};
  \node[mycircle,below left= 1.6cm ](c5) at (c1){$\sigma(1)$};
 \node[mycircle,above left = 1.6cm](c6) at (c4){$1$};
\node[state,below =2.7cm](c7) at (c1){};
  \node[state,right =3cm](c8) at (c7){};
  \node[state,right =1cm](c9) at (c8){};
\draw[myarrow] (c2) to node {} (c3);
\draw[dashed,->](c1)  [bend left] edge node[above]{}(c2);
\draw[dashed,->](c3)  [bend left] edge node[above]{}(c4);
\draw[dashed,->](c4)  [bend left] edge node[above]{}(c6);
\draw[myarrow](c6) to node[above]{}(c5);
\draw[dashed,->](c5)  [bend left] edge node[above]{}(c1);
\draw[->](c8)--(c9);
\node[state,right =8cm](c10) at (c1){};
\node[state,below =5.4cm](c11) at (c10){};
\node[mycircle, below right=1.8cm] (c12) at (c10) {$k$};
  \node[mycircle,above right = 1.8cm] (c13) at (c11) {$\sigma(k)$};
  \node[mycircle,below left= 1.8cm ](c14) at (c10){$\sigma(1)$};
  \node[mycircle,above left = 1.8cm](c15) at (c11){$1$};

\draw[myarrow] (c12) to  node {} (c14);
\draw[dashed,->](c10)  [bend left] edge node[above]{}(c12);
\draw[dashed,->](c14)  [bend left] edge node[above]{}(c10);
\draw[dashed,->](c13)  [bend left] edge node[above]{}(c11);
\draw[myarrow](c15)   edge node[above]{}(c13);
\draw[dashed,->](c11)  [bend left] edge node[above]{}(c15);

\end{tikzpicture}}
\caption{The change in the cycle structure of $\sigma$, when $\sigma(1)$ 
is swapped with $\sigma(k)$, when $1$ and $k$ belong to the same cycle.}
\label{cycle-break}
\end{figure}
\begin{figure}[htp]
\centering
\scalebox{.67}{\begin{tikzpicture}[
      mycircle/.style={
         circle,
         draw=black,
         fill=gray,
         fill opacity = 0.3,
         text opacity=1,
         inner sep=0pt,
         minimum size=25pt,
         font=\small},
         state/.style={circle,inner sep=0pt, minimum size=2pt},
      myarrow/.style={-Stealth},
      node distance=0.6cm and 1.2cm
      ]

\node[state](c1){};
\node[state,below=2cm](c2) at (c1){};
\node[state,below=2cm](c3) at (c2){};
\node[state,below left=1.2cm](c4) at (c1){};
\node[state,above left=1.2cm] (c5) at (c3){};
\node[mycircle,below right=.8cm](c6) at (c1){$1$};
\node[mycircle,above right=.8cm](c7) at (c3){$\sigma(1)$};
\draw[dashed,-Stealth](c1)  [bend left] edge node[above]{}(c6);
\draw[dashed,-Stealth](c4)  [bend left] edge node[above]{}(c1);
\draw[dashed,-Stealth](c5)  edge node[above]{}(c4);
\draw[dashed,-Stealth](c7)  [bend left] edge node[above]{}(c3);
\draw[dashed,-Stealth](c3)  [bend left] edge node[above]{}(c5);
\draw[myarrow] (c6) to  node {} (c7);

\node[state,right=2cm](c8) at (c2){+};

\node[state,right=4.5cm](c9) at (c1){};
\node[state,below=2cm](c10) at (c9){};
\node[state,below=2cm](c11) at (c10){};
\node[mycircle,below left=.8cm](c12) at (c9){$\sigma(k)$};
\node[mycircle,above left=.8cm](c13) at (c11){$k$};
\node[state,below right=1.2cm](c14) at (c9){};
\node[state,above right=1.2cm](c15) at (c11){};

\draw[dashed,-Stealth](c9)  [bend left] edge node[above]{}(c14);
\draw[dashed,-Stealth](c14)  edge node[above]{}(c15);
\draw[dashed,-Stealth](c15)  [bend left] edge node[above]{}(c11);
\draw[dashed,-Stealth](c11)  [bend left] edge node[above]{}(c13);
\draw[dashed,-Stealth](c12)  [bend left] edge node[above]{}(c9);
\draw[myarrow] (c13) to node {} (c12);

\node[state,right=2cm](c16) at (c10){};
\node[state,right=1cm](c32) at (c16){};
\draw[->](c16)--(c32);

\node[state,right=5.5cm](c17) at (c9){};
\node[state,below=2cm](c18) at (c17){};
\node[state,below=2cm](c19) at (c18){};
\node[mycircle,below right=.8cm](c20) at (c17){$1$};
\node[mycircle,above right=.8cm](c21) at (c19){$\sigma(1)$};
\node[state,below left=1.2cm](c22) at (c17){};
\node[state,above left=1.2cm](c23) at (c19){};

\draw[dashed,-Stealth](c17)  [bend left] edge node[above]{}(c20);
\draw[dashed,-Stealth](c22)  [bend left] edge node[above]{}(c17);
\draw[dashed,-Stealth](c23)  edge node[above]{}(c22);
\draw[dashed,-Stealth](c19)  [bend left] edge node[above]{}(c23);
\draw[dashed,-Stealth](c21)  [bend left] edge node[above]{}(c19);

\node[state,right=4.5cm](c24) at (c17){};
\node[state,below=2cm](c25) at (c24){};
\node[state,below=2cm](c26) at (c25){};
\node[mycircle,below left=.8cm](c30) at (c24){$\sigma(k)$};
\node[mycircle,above left=.8cm](c29) at (c26){$k$};
\node[state,below right=1.2cm](c27) at (c24){};
\node[state,above right=1.2cm](c28) at (c26){};

\draw[dashed,-Stealth](c24)  [bend left] edge node[above]{}(c27);
\draw[dashed,-Stealth](c27)  edge node[above]{}(c28);
\draw[dashed,-Stealth](c28)  [bend left] edge node[above]{}(c26);
\draw[dashed,-Stealth](c26)  [bend left] edge node[above]{}(c29);
\draw[dashed,-Stealth](c30)  [bend left] edge node[above]{}(c24);
\draw[myarrow] (c20) to node {} (c30);
\draw[myarrow] (c29) to  node {} (c21);

\end{tikzpicture}
	}
\caption{The change in the cycle structure of $\sigma$ when $\sigma(1)$ 
is swapped with $\sigma(k)$, when $1$ and $k$ belong to different cycle.}
\label{cycle-merge}
\end{figure}
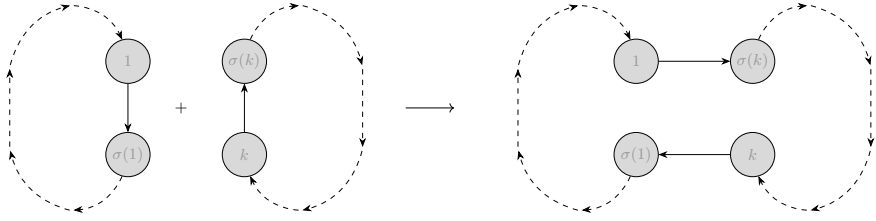
Notice that, if $\sigma(1)=k$, then 
one of the resulting cycles is a singleton. In the case when $1$ and $k$ belong
to different cycles of $\sigma$, these two cycles merge and form a single cycle
in $\pi$ (Fig.~\ref{cycle-merge}).

Given two permutations $\pi$ and $t$, we call the cycles of $\pi \circ t^{-1}$
as the cycles of $\pi$ {\em relative} to $t$.  The above observations about the
cycle structure of two permutations $\pi$ and $\sigma$ which differ by a single
swap between $1$ and $k$ will apply in this case to the relative cycle
structure of $\pi$ and $\sigma$ with respect to $t$.

\subsection{Routing in undirected star graph}
\label{secUndirectedRouting}

In this section, we describe the routing algorithm for the undirected star
graph  $\SN$ presented in \cite{akers1989group}. Assume that a node labelled
$c$ forwards a packet $P$ from a source $s$ to a destination $t$. The
destination label $t$ is available in the packet.  Upon receiving the packet,
$c$ accepts $P$ if $c$ is same as $t$.  Otherwise, when $c(1) \neq t(1)$, the
node $c$ forwards $P$ through the link labelled $i$, where $i$ is the position of
$c(1)$ in $t$.  We call such a move a {\em settling move}. A value is called
{\em settled} if it is in the same location in $c$ and $t$, and {\em unsettled}
otherwise. When $c(1)=t(1)$ (but $c \neq t$), the node $c$ forwards $P$ through
a link $i$, where $i$ is the position of an unsettled value. We call such move
a {\em seeding move}.  Notice that, during the course of routing $P$ from $s$
to $t$, the number of seeding moves is same as the number of non-singleton
cycles in $s$ relative to $t$. Also, no move disturbs an already settled value.
Therefore, we can observe that the total number of steps required to settle all
unsettled values in $s$, denoted by $d$, is at most $m(s,t)+c(s,t)$, 
where $m(s,t)$ is the number of mismatched values (\ie values that are not in 
their correct position with respect to $t$) and $c(s,t)$ is the number of 
non-singleton cycles in $s$ relative to $t$.  More closer analysis yield 
\cite{akers1989group} the following result.
\begin{equation}\label{eqnUndirected}
d =
    \begin{cases}
      m(s,t)+c(s,t), & \text{if}\ s(1)=t(1)\\
      m(s,t)+c(s,t)-2, & \text{otherwise}.
    \end{cases}
\end{equation}
It is not difficult to argue that the above routing algorithm is 
optimal and hence $d$ is the distance between $s$ and $t$ in $\SN$.
\section{The proposed routing algorithm}
There are two different ways in which one can describe and analyse a routing
algorithm on a star graph.  In the first view, which we call the ``network
view'', we consider each vertex of $\SN$ as a communication node whose {\em
address} is the permutation labelling that vertex.  Depending on the sign of
the address of a node, we classify it as an {\em even node} or an {\em odd
node}.  We consider each arc of $\USN$ as a unidirectional communication link
and label it by the unique position in $\{2, \ldots, n\}$, where the addresses
of the endpoints of the arc differ. Hence every node has $n-1$ links attached
to it with unique labels from $\{2, \ldots, n\}$. For an even node, the links
labelled $2$ to $\ceil{(n-1)/2} + 1$ are outgoing links and the remaining are
incoming. The situation is reversed for odd nodes. Every packet that is to be
routed along the network will have the destination address in its
header. We describe the algorithm by which a node, on receiving a packet not
destined for itself, selects the outgoing link along which to relay that
packet.  This selection is based on the address of the current node and the
destination address. 

In the second view, which we call the ``sorting view'', we consider each vertex
of $\SN$ as a permutation of $[n]$. Thus a routing is viewed as a step-by-step
procedure to sort the permutation labelling the source to the permutation
labelling the destination. Each step in this sorting is restricted to be a
transposition $(1, i)$, where $i \in \{2, \ldots, \ceil{(n-1)/2} + 1\}$, if the
current permutation is even, and $i \in \{\ceil{(n-1)/2}+2, \ldots, n\}$, if
the current permutation is odd. Hence a directed path in $\USN$ will correspond
to an alternating sequence of right half and left half transpositions.  This is
the view with which we will analyse our routing algorithm in
Section~\ref{secAnalysis}.

In a given permutation, let us call the positions $2$ to $\ceil{(n-1)/2} + 1$
as the {\em left half}, and the positions $\ceil{(n-1)/2} + 2$ to $n$ as the 
{\em right half}. First, we analyse the case of sorting a permutation $\pi$ in 
which all the left values are in a derangement in the left half itself, and all 
the right values are in a derangement in the right half itself. For every $n 
\geq 5$, an example for $\pi$ is the permutation obtained by cyclically 
shifting the left-half and right-half by one position each. That is, the cycle 
decomposition of $\pi$ is $(1)(2, \ldots, k)(k+1, \ldots, n)$, where $k =
\ceil{(n-1)/2}+1$.  This analysis serves two purposes. Firstly, it establishes 
a lower bound on the diameter of $\USN$. Secondly, it illustrates a typical run 
of our proposed algorithm to be described later. Let $\pi = \pi^0, \ldots, 
\pi^l = id$ be the nodes of a shortest directed path from $\pi$ to $id$ in 
$\USN$. Notice that, in $\pi$, every value except $1$ is not in its ``correct'' 
position (with respect to the identity permutation) and hence needs to be 
moved. This requires a transposition $a = (1, \pi^{-1}(i))$ to remove $i$ from 
its present position and a transposition $b = (1, i)$ to place $i$ in its final 
position. Let $\alpha$ and $\beta$ be, respectively, the permutations in 
$\{\pi^1, \ldots, \pi^l\}$ which appear immediately after the transposition $a$ 
and immediately before the transposition $b$. Notice that $\alpha(1) = \beta(1) 
= i$.  The key observation is that $\alpha$ and $\beta$ cannot be the same 
permutation. This is because, for every $i \in \{2, \ldots, n\}$, both 
$\pi^{-1}(i)$ and $i$ are in the same half and the directions in $\USN$ 
constraints one to alternate between left half and right half transpositions. 
Hence for every $i \in \{2, \ldots, n\}$, there exists at least two distinct 
permutations in $\{\pi^0, \ldots, \pi^l\}$ which has $i$ in the first position. 
Moreover $\pi^0(1) = \pi^l(1)=1$ (\ie. the value 1 appears in the first position for 
at least two permutations). Thus $l+1 \geq 2n$ and hence the length of the path 
is at least $2n-1$. If $\pi$ was an even permutation, we could have improved 
the lower bound by $1$, since the distance between two even permutations has to 
be even.  This is indeed the case when $n$ is odd. When $n$ is even and $n \geq 
8$ (and thereby $k \geq 5$, we can choose $\pi$ to be 
$(1)(2,3)(4,\ldots,k)(k+1,\ldots, n)$. This improvement does not work for $n = 6$, and it is indeed established by computer simulation that $\diam{\USN} = 2n-1$ when $n = 6$ \cite{day1993unidirectional}. Hence we conclude
 \begin{thm}\label{thmLowerBound}
For every $n \geq 5$ the diameter of $\USN$ is at least $2n-1$. Further if
$n \neq 6$, the diameter of $\USN$ is at least $2n$.
\end{thm}

Now let us see a way to sort the permutation $\pi =
(1)(2,\ldots,k)(k+1,\ldots,n)$ for an odd $n \geq 5$ and $k = \ceil{(n-1)/2} + 1$. We do not attempt to rigorously justify the claims made in the following discussion as they are proved in more generality in Section~\ref{secAnalysis}. We do the sorting in two phases. In the first phase (the {\em crossing phase}), we obtain a permutation $\gamma$ in which all the values in $\{2, \ldots, k\}$ (the {\em small values}) are in the right half and all the values in $\{k+1, \ldots, n\}$ (the {\em big values}) are in the left half.  This can be done in $n+1$ steps; the first step places $1$ in the left half (seeding move) and all the subsequent steps either places a small value in the right half or a large value in the left half (crossing moves). Only thing one has to be careful about is to remove $1$ from the
left half only in the last transposition. For example, one can attack the
positions $2,k+1,3,k+2,\ldots,k-1,n,2$ in that order to arrive at $\gamma =
(1)(2,k+1)(3,k+2)\cdots(n-1,n)$. In the second phase (the {\em settling
phase}), when a value $i$, $i \neq 1$ appears in the first position for the
first time, in the very next step we will settle it, i.e., place it in position $i$.  This will be possible since, $\gamma^{-1}(i)$ and $i$ are in different halves for all $i \in {2, \ldots.  n}$. This phase could have been completed in $n$ steps provided the elements $\{2, \ldots, n\}$ formed a singe cycle in $\gamma$. Otherwise, after placing all the elements in a cycle of $\gamma$ to their correct positions, $1$ will return to the first position.  This results in one extra move (a {\em seeding move}) per non-singleton cycle of $\gamma$. An extreme example of this can be seen by analysing the case when $\gamma$ is as above, wherein one requires $\ceil{n/2}-1$ seeding moves.  Hence the number of moves in the settling phase is $n-1 + c(\gamma)$ where $c(\gamma)$ is the number of non-singleton cycles in $\gamma$.  Since the number of non-singleton cycles in any permutation of $[n]$ is at most $\floor{n/2}$, one quickly sees that $\pi$ can be sorted in a total of $\floor{5n/2}$ steps. One can then easily extend this analysis to an arbitrary permutation in place of $\pi$ and show that the diameter of $\USN$ is at most $5n/2 + O(1)$, reproving the bound of Fujita \cite{fujita2013oriented}. But we show that we have enough freedom while building $\gamma$ to ensure that $\gamma$ consists of at most two non-singleton cycles. This is done by showing that during all but the final two transpositions of the crossing phase,	we can select the swaps so as not to complete a new cycle among the crossed values. This is what helps us in achieving the bound of $2n + O(1)$ on the diameter of $\USN$.

One drawback of the above method is that, even if the source permutation $\pi$ is very close to the identity permutation in terms of distance in $\SN$ like $\pi = (1)(2,3)(4)(5)\cdots(n)$, this method may take $2n$ steps. Hence, we
modify the above method by making sure that, if $\pi$ has $m$ small values and
$m$ large values which are already in their correct positions, then those $2m$
values are not disturbed during the sorting. We then analyse this strategy to
show that any permutation $\pi$ can be sorted in at most $4d+4$ steps, where
$d$ is the distance between $\pi$ and $id$ in $\SN$.

These attempts to reduce the number of cycles in $\gamma$ and to disturb as few
settled values in $\pi$ as possible is what makes the crossing phase of the
routing algorithm slightly complex.  Moreover, when $\pi(1) \neq 1$, we have
two possibilities.  If $\pi$ is even and $\pi(1)$ is a large value, we continue
as if we are in the crossing phase. We do the same when $\pi$ is odd and
$\pi(1)$ is a small value. In the other two cases ($\pi$ even, $\pi(1)$ small
and $\pi$ odd, $\pi(1)$ large), we start by settling $\pi(1)$ and continue in
the settling phase till either $1$ appears in the first position or one of the
two cases mentioned above occurs.  Then we go into the crossing phase, complete
it, and enter the settling phase for a second time. Hence one cycle of the
settling phase can happen before the crossing phase. With this high-level idea,
we formally state our proposed routing algorithm.

\begin{defn}\label{DataStructures}
For a permutation $s$ of $[n]$, we call 
$L(s) = \{s(i): 2 \leq i \leq \ceil{(n-1)/2} + 1\}$ and
$R(s) = \{s(i): \ceil{(n-1)/2} + 2 \leq i \leq n\}$ as the 
sets of {\em left values} and {\em right values} of $s$, respectively.

Given two permutations $s$ and $t$ of $[n]$ for some $n$,
we define $S(s,t) = \{s(i) : s(i) = t(i), 1 \leq i \leq n \}$,
and $U(s,t) = [n] \setminus S(s,t)$ respectively, as the sets of {\em settled} 
and {\em unsettled} values between $s$ and $t$. We partition 
$U(s,t) \setminus\{s(1), t(1)\}$ into four sets
\begin{align*}
	ULL(s,t) = U(s,t) \cap L(s) \cap L(t), \\
	URR(s,t) = U(s,t) \cap R(s) \cap R(t), \\
	ULR(s,t) = U(s,t) \cap L(s) \cap R(t), \\
	URL(s,t) = U(s,t) \cap R(s) \cap L(t).
\end{align*}
We also partition $S(s,t)$ into two sets
\begin{align*}
	SL(s,t) = S(s,t) \cap L(t), \quad&
	SR(s,t) = S(s,t) \cap R(t).
\end{align*}
Let us call $X(s,t) = ULR(s,t) \cup URL(s,t)$ as the set of
{\em crossed values} between $s$ and $t$.
A cycle of $s$ relative to $t$ is called {\em alternating} if
it has size at least two, and the successive elements of the cycle alternate between $L(s)$ and $R(s)$. 
Finally, $\chi(s,t)$ will denote the number of
alternating cycles of $s$ with respect to $t$.
\end{defn}
The processing done by an even node is given in Algorithm~\ref{RouteEven}. 
	The processing done by an odd node is similar (the roles of ``left'' and ``right'' are reversed) and hence omitted. In every 
	move, $c(1)$ is exchanged with $c(i)$ for some $i \in \{2, \ldots, n\}$.  
	We classify these moves into three types. 
	If $c(1) = t(1)$, the move is called a {\em seeding move}. 
	If $c(1) \neq t(1)$  and $c(1) = t(i)$, i.e., 
	$c(1)$ moves to its correct location in $t$, 
	it is called a {\em settling move}. 
	If $c(1) \in L(t)$ and it moves to the right half or 
	if $c(1) \in R(t)$ and it moves to the left half, 
	the move is called a {\em crossing move}.
\begin{algorithm*}[ht]
\begin{algorithmic}[1]
\caption{Processing done by an even node labelled $c$ 
	upon receiving a packet $P$ destined for a node labelled $t$.}
  \Procedure{RouteEven(Packet $P$)}{}
  	\State Receive packet $P$, extract the destination address $t$.
  	\State If the address of the current node $c$ is the same
      as $t$, accept $P$ and return.
  	\State Let $L_t=L(t)$, $R_t=R(t)$, $ULL=ULL(c,t)$, $URR=URR(c,t)$, 	
  		$SL=SL(c,t)$, $URL=URL(c,t)$.
  	\State {\em Case 1 (Settling Move):} $c(1) \in L_t$.
  	\State \hspace{.2cm} Let $i$ be the position of $c(1)$ in
   		the permutation $t$.
  	\State {\em Case 2 (Crossing/Seeding Move):} $c(1) \not\in L_t$ and 
  	$\card{ULL} + \card{URR} >0$.
    \Statex\hspace{.9cm}({\em Crossing} when $c(1)\in R_t$, {\em Seeding} 	
    	when $c(1)=t(1)$)
  	\State \hspace{.2cm}The forwarding link $i$ is selected based on the 
  	cycle structure of $c$ with respect to $t$. 
  		\State \hspace{.2cm}{\em Case 2.1:} When $ULL$ contains a value 
  			that is not part of the cycle containing $c(1)$.
		  \State \hspace{.5cm} Pick $i$ as the $c$-index of that value.
		\State \hspace{.2cm}{\em Case 2.2:} When all values in $ULL$ are 
			part of the cycle containing $c(1)$.
		  \State  \hspace{.5cm} Pick $i$ as the $c$-index of the value in 
		  	$ULL$ that comes first on traversing this cycle backward 
		  	from $c(1)$.
		\State  \hspace{.2cm}{\em Case 2.3:}When $ULL$ is empty
		  \State  \hspace{.5cm}Pick $i$ as the $c$-index of	any value
   			from $SL$.
   \State {\em Case 3(Crossing Move):} $c(1) \in R_t$ and 
   		$\card{ULL} + \card{URR} =0$.
  		\State  \hspace{.2cm}{\em Case 3.1(Final Crossing Move):}
  			When $\card{ULR} > 0$
  			\State \hspace{.5cm} Pick $i$ as the $c$-index of a value from  
  				$ULR$. If possible, select $i$ from an alternating cycle 
  				in $c$.
	  	\State  \hspace{.2cm}{\em Case 3.2(Final/Pre-Final Crossing Move):}
	  		When $\card{ULR} = 0$
		  \Statex\hspace{.9cm}({\em Final Crossing Move} when $t(1)$ is 
		  	picked, {\em Pre-Final Crossing Move} when a settled value 
		  	is picked)
		  	\State  \hspace{.5cm} Pick $i$ as the $c$-index of $t(1)$ if 
		  		possible, otherwise pick a settled value.
  	\State {\em Case 4 (Seeding Move):} $c(1) = t(1)$ and 
	 $\card{ULL}+ \card{URR}=0$.
		\State \hspace{.2cm} Pick $i$ as the $c$-index of a value from 
			$ULR$.	
	\State Send $P$ along the edge labelled $i$ and terminate.	
 \EndProcedure
\end{algorithmic}
\addtocounter{algorithm}{-1}
\caption{Processing done by an even node labelled $c$ 
	upon receiving a packet $P$ destined for a node labelled $t$.}\label{RouteEven}
\end{algorithm*}
\section{Analysis of the proposed routing algorithm}
\label{secAnalysis}
In this section, we are going to prove the following upper bound on the number
of hops that Algorithm~\ref{RouteEven} uses to reach from a node $s$ to a node
$t$ based on the relative structure of the permutations that label $s$ and $t$.
\begin{thm} \label{thmMain}
Let $s$ and $t$ be the permutations labelling any two nodes of the 
oriented star graph $\orient{\SN}$. Then, Algorithm~\ref{RouteEven} will 
send a packet from $s$ to $t$ in at most
\[
\card{X(s,t)}+\max\{6, y\}
\] steps, where
\[
	y=4 \max \{\card{ULL(s,t)}, \card{URR(s,t)} \} +
	 \chi(s,t) + 4.
\]
\end{thm}
\begin{proof}
Even though we have presented Algorithm~\ref{RouteEven} as a distributed
routing protocol between nodes, we will present the analysis as sequence of
swaps on permutations, starting with the permutation labelling the node $s$ and 
ending with the permutation labelling the node $t$. We will denote these permutations 
also by $s$ and $t$ respectively. Moreover, all the swaps will be between the first 
element of a permutation $\pi$  and an element from the left (resp., right) half of 
$\pi$,  if $\sign(\pi)$ is even (resp., odd).  Thereby, we make sure that every swap 
corresponds to a directed edge in $\orient{\SN}$.

We group the sequence of moves into three (possibly empty) phases, {\em Phase
One}, {\em Phase Two}, and {\em Phase Three} in that order.  Phase One is
non-empty only if the first move is a settling move, in which case, Phase One
consists of all the settling moves before any seeding or crossing move. Phase
Two will be non-empty only if there is at least one crossing move, in which
case, it starts after Phase One and continues till Final Crossing move. All the
remaining moves are called Phase Three moves (Fig.~\ref{anaysis}).  We call
the permutation at the end of the Phase One as $\alpha$, at the end of the
Phase Two as $\gamma$ and the one that before $\gamma$ as $\beta$.
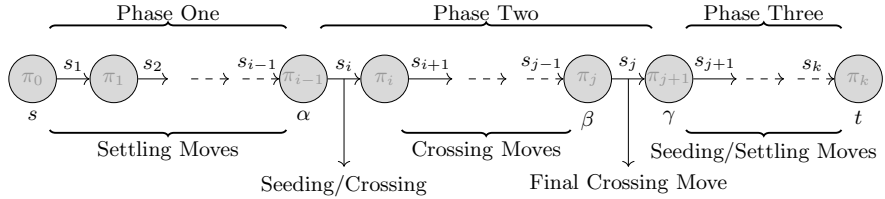
\begin{figure*}[htp]
\centering
	\scalebox{.9}{
%
%
%

\begin{tikzpicture}
[
	mycircle/.style={
         circle,
         draw=black,
         fill=gray,
         fill opacity = 0.3,
         text opacity=1,
         inner sep=0pt,
         minimum size=20pt,
         font=\small},
	state/.style={circle,inner sep=0pt, minimum size=2pt},
    node distance=1.2cm
]
        
  \node[mycircle,label=below:$s$] (p0) {$\pi_{0}$};
  \node[mycircle, right of=p0 ] (p1) {$\pi_{1}$};

  \node[state, right of=p1,node distance = 1.1cm](p2){};
   \node[state, right of=p2,node distance = 1.1cm](p3){};

  \node[mycircle, label,right of= p3,node distance = .6cm,label=below:$\alpha$](p4) {$\pi_{i-1}$};
  \node[mycircle, label,right of= p4](p5) {$\pi_{i}$};
   \node[state, right of=p5](p9){};
   \node[state, right of=p9](p10){};
  \node[mycircle, label,right of= p5,node distance = 3cm,label=below:$\beta$](p6){$\pi_{j}$};
  
  \node[mycircle, label,right of= p6,node distance = 1.2cm,label=below:$\gamma$](p7){$\pi_{j+1}$};
   \node[state, right of=p7,node distance = 1.1cm](p11){};
   \node[state, right of=p11,node distance = 1.1cm](p12){};
  
  \node[mycircle, label,right of= p12,node distance = .6cm,label=below:$t$](p8){$\pi_k$};
  

  \draw[->] (p1) edge node[above] {$s_2$} +(.8,0);
  \draw[->](p0) edge node[above]{$s_1$}(p1);
  
  \draw[dashed,->](p2) -- ($(p2) !.6cm! (p3)$);
  
  \draw[dashed,<-] (p4) -- node[above] {$s_{i-1}$} ++(-1cm,0);
  
  \draw[->](p4)-- node[above]{$s_i$}(p5);
  \draw[->] ($(p4.east)!0.5!(p5.west)$) -- ++(0,-1.3cm) node[below] {Seeding/Crossing};
  \draw[->] (p5) edge node[above] {$s_{i+1}$} +(1,0);
  
   \draw[dashed,->](p9) -- ($(p9) !.6cm! (p10)$);
\draw[dashed,<-] (p6) -- node[above] {$s_{j-1}$} ++(-1cm,0);
 \draw[->](p6)-- node[above]{$s_j$}(p7);
\draw[->] ($(p6.east)!0.5!(p7.west)$) -- ++(0,-1.3cm) node[below] {Final Crossing Move};

  \draw[->] (p7) edge node[above] {$s_{j+1}$} +(1,0);
  \draw[dashed,->](p11) -- ($(p11) !.6cm! (p12)$);
\draw[dashed,<-] (p8) -- node[above] {$s_k$} ++(-1cm,0);

\draw [
    thick,
    decoration={
        brace,
        mirror,
        raise=0.5cm
    },
    decorate
] (p0.south east) -- (p4.south west) 
node [pos=0.5,anchor=north,yshift=-0.55cm] {Settling Moves}; 
\draw [
    thick,
    decoration={
        brace,
        mirror,
        raise=0.5cm
    },
    decorate
] (p5.south east) -- (p6.south west) 
node [pos=0.5,anchor=north,yshift=-0.55cm] {Crossing Moves}; 
\draw [
    thick,
    decoration={
        brace,
        mirror,
        raise=0.5cm
    },
    decorate
] (p7.south east) -- (p8.south west) 
node [pos=0.5,anchor=north,yshift=-0.55cm] {Seeding/Settling Moves};  

\draw [
    thick,
    decoration={
        brace,
        raise=1cm
    },
    decorate
] (p0.south east) -- (p4.south west) 
node [pos=0.5,anchor=south,yshift=1cm] {Phase One}; 
\draw [
    thick,
    decoration={
        brace,
        raise=1cm
    },
    decorate
] (p4.south east) -- (p7.south west) 
node [pos=0.5,anchor=south,yshift=1cm] {Phase Two}; 
\draw [
    thick,
    decoration={
        brace,
        raise=1cm
    },
    decorate
] (p7.south east) -- (p8.south west) 
node [pos=0.5,anchor=south,yshift=1cm] {Phase Three};

\end{tikzpicture}
}
	\caption{A pictorial representation of the analysis of Algorithm~\ref{RouteEven}}
	
  \label{anaysis}
\end{figure*}

\begin{obs} \label{obs}

All the moves in Phase One are settling moves and all the moves in 
Phase Two except possibly the first are crossing moves.
\end{obs}

\begin{proof}
It follows from the definition of phases that all the moves in Phase One are settling moves. As long as $\card{ULL} + \card{URR} > 0$, any crossing move from an even permutation will pick a left value, and from an odd permutation will pick a right value.  This is ensured by Case 2 of the algorithm. Hence the next move will also be a crossing move.  Once $\card{ULL} + \card{URR} = 0$, we are in Case 3 of the algorithm.  In Case 3.1, the algorithm performs a crossing move by picking an already crossed value. Let us call this move as $s_j$. The next move will be a settling move. Since $\card{ULL} + \card{URR}$ does not increase due to any move of the algorithm (which is ensured by Case 1 of the algorithm), $\card{ULL} + \card{URR}$ remains zero for the rest of the moves. Hence, Case 2 of the algorithm will never occur beyond $s_j$. Also, after $s_j$, since the set $URR$ is empty, a value from the set $R_t$ will never become the first value in an even permutation. Similarly the set $ULL$ is empty after $s_j$, a value from the set $L_t$ will never become the first value in an
odd permutation. Therefore, Case 3 of the algorithm never happens beyond $s_j$. So, Phase Two of the algorithm consists of entirely crossing moves, except possibly the first one and hence Observation~\ref{obs}.  
\end{proof} 

\begin{prop}\label{propPhaseOne} 
The number of crossed value at the end of Phase One, $\card{X(\alpha,t)} = \card{X(s,t)} -(i-1)$ and the number of alternating cycles in $\alpha$ is at most the number of alternating cycles in $s$, \ie $\chi(\alpha,t)\leq \chi(s,t)$.
\end{prop}

\begin{proof}
By Observation~\ref{obs}, in Phase One, every move is a settling move.  When a move is a settling move, its previous move swaps the value in the first position with a crossed value. Hence, during the Phase One, the number of crossed values come down by $(i-1)$.
Every settling move splits an existing cycle into a singleton cycle and another cycle which contains the first value. Since neither of the above, is an alternating cycle $\chi(\alpha,t)\leq \chi(s,t)$. 
\end{proof}

\begin{prop}\label{propPhaseTwo}
The number of alternating cycles at the end of Phase Two,
\begin{equation}\label{eqnAlternatingCyclesGamma}
\chi(\gamma,t)\leq
	\begin{cases}
		 1, \ \text{when } ULL(\alpha,t) = URR(\alpha,t) = ULR(\alpha,t) = \emptyset \\
		 \chi(\alpha,t)+1, \	\text{otherwise.}
	\end{cases}
\end{equation}

Also, the number of moves in Phase Two,
 \begin{equation}\label{eqnUndirectedRdf}
 m_2 \leq
    \begin{cases}
     2, \ \text{when } ULL(\alpha,t) = URR(\alpha,t) = ULR(\alpha,t) = \emptyset\\
	2 \max\{\card{ULL(\alpha,t)}, \card{URR(\alpha,t)}\}+1, \			\text{otherwise.}
\end{cases}
\end{equation}
Furthermore, the number of crossed values in $\gamma$, 
\begin{equation}\label{eqnCrossedValues}
\card{X(\gamma,t)} \leq
\begin{cases}
\card{X(\alpha,t)} +2,  \text{when } ULL(\alpha,t) = URR(\alpha,t) =  ULR(\alpha,t) = \emptyset\\
\card{X(\alpha,t)} + 2 \max\{\card{ULL(\alpha,t)}, \card{URR(\alpha,t)}\},\ \ \text{otherwise.}
\end{cases}
\end{equation}
\end{prop}

\begin{proof}
At first, let us prove the claim on the number of alternating cycles at the end
of Phase two, \ie $\chi(\gamma,t)$. In Phase Two, when $\card{ULL(\alpha,t)} +
\card{URR(\alpha,t)} = 0$ and $\card{ULR(\alpha,t)} = 0$, there will be at
least one crossing move. This move will be carried out with either $t(1)$ or a
settled value (Case 3.2). If it is carried out with $t(1)$, then it will be the
final crossing move. Otherwise, the very next move will be carried out with
$t(1)$ and will represent the final crossing move. In both cases,
$\gamma(1)=t(1)$. Also, there will be exactly one crossed value in both halves
of $\gamma$. Since all the remaining values are settled values, the crossed
values in $\gamma$ will form an alternating cycle. Hence $\chi(\gamma,t)=1$. 

Now, let us consider the case when $ ULL(\alpha,t)$ or $URR(\alpha,t)$ or $
ULR(\alpha,t)$ is non-empty. We will first show that, the number of alternating
cycles before the final crossing move in Phase Two is at most two more than the
number of alternating cycles at the end of Phase One, \ie $\chi(\beta,t)\leq
\chi(\alpha,t)+2$.

In a crossing move of an even permutation, if the exchange happens either with
a settled left value or an unsettled left value that does not belong to the
cycle containing $1$, this will result in a merging of two cycles into a new
cycle. Since this new cycle contains $1$, it will not be an alternating cycle
(Section B). If the exchange takes an unsettled left value from a cycle
containing $1$, this will break the cycle into two.  One of which contains $1$
and hence not an alternating cycle. The second one will not be an alternating
cycle if it further contains any other unsettled left value. This means the
only case in which a new alternating cycle may be created is the crossing move
from the even permutation with $\card{ULL} =1$ and the crossing move from the
odd permutation with $\card{URR}=1$. Hence $\chi(\beta,t)\leq
\chi(\alpha,t)+2$.

If there is an alternating cycle before the Final Crossing Move (\ie in the
permutation $\beta$), the Final Crossing Move exchanges an already crossed
value in this cycle (Case 3.1), resulting in the merging of the two cycles. This
will reduce the number of alternating cycles in $\gamma$ by one. Hence
$\chi(\gamma,t)\leq \chi(\beta,t)-1=\chi(\alpha,t)+1$.

Next we prove the claim on the number of moves in Phase Two. If $ULL(\alpha,t),
 URR(\alpha,t)$ and $ULR(\alpha,t)$ are empty, 
we have one crossing move, if $d(1)\in L_\alpha$ and two otherwise (Case 3.2).

When $ULL(\alpha,t)$ and $URR(\alpha,t)$ are empty but $\card{ULR(\alpha,t) } >0$, we have 
one crossing move (Case 3.1). If either $ULL(\alpha,t)$ or $URR(\alpha,t)$ is non-empty, 
the larger of the two reduces by one in every alternate crossing move of Case 2 of Algorithm~\ref{RouteEven}. 
If at least two crossing moves of Case 2 happen, Case 3.2 will not occur. Hence we can conclude that,
the number of crossing moves is at most $2 \max\{\card{
ULL(\alpha)},\card{URR(\alpha)}\}+1$, when $\card{ULL(\alpha)} + \card{
URR(\alpha)} >0$.

Lastly, we prove the claim on the number of crossed values after Phase Two. 
Since every move of Phase Two adds at most one element to either $ULR$ or
$URL$, $\card{X(\gamma,t)} \leq \card{X(\alpha,t)}+m_2$.  Further, either $ULR(\alpha, t) \neq \emptyset$ or 
if at least two crossing moves of Case 2 happen, Case 3.2 will not occur. Hence the final 
crossing move (Case 3.1) will only replace one crossed value with another.
\end{proof}

\begin{prop}\label{propPhaseThree}
The number of moves in the Phase Three is equal to $\card{X(\gamma,t)}+c(\gamma,t)$.
\end{prop}
\begin{proof}
In \cite{akers1989group} Akers and Murthy proved that the minimum number of steps required to reach a node, $t$, from a node $s$ in $\SN$, denoted by $d$ is
\begin{equation}\label{eqnUndirectedRdfTwo}
	d = \begin{cases}
		m(s,t)+c(s,t), & \text{if}\ s(1)=t(1)\\ m(s,t)+c(s,t)-2, & \text{otherwise.}
	\end{cases} 
\end{equation}
 where $m(s,t)$ is the number of mismatched values
(\ie the values that are not in their correct position) in $s$ and $c(s,t)$ is
the number of cycles in $s$ with respect to $t$. But we know that
\begin{equation}\label{eqnUndirectedRdfUnsettled}
m(s,t) =
\begin{cases}
      \card{ULL(s,t)}+ \card{URR(s,t)} +\card{X(s,t)}, \text{if}\ s(1)=t(1).\\
      \card{ULL(s,t)} + \card{URR(s,t)} +\card{X(s,t)} +2, \text{otherwise.}
    \end{cases}
\end{equation}
Therefore we can write Equation~\ref{eqnUndirectedRdfTwo} as,
\begin{equation}\label{eqnUndirectedDistance}
d = \card{ULL(s,t)} + \card{URR(s,t)} +\card{X(s,t)} + c(s,t).
\end{equation}
  Since $\card{ULL(\gamma,t)}=\card{URR(\gamma,t)} =0$, using 
  Equation~\ref{eqnUndirectedDistance}, we can conclude that the 
  number of moves in the Phase Three is equal to $\card{X(\gamma,t)}+c(\gamma,t)$.
\end{proof}
Let $i$ be the number of steps taken by the algorithm in Phase One. 
Using propositions \ref{propPhaseTwo} and \ref{propPhaseThree} we can bound the
total number of steps, $\orient{d}$, taken by the algorithm as
	\begin{equation}\label{TheoremThreeProof}
		\orient{d} \leq
		\begin{cases}			
			i + 2+(\card{X(\gamma,t)}+c(\gamma,t)), \, \text{when } ULL(\alpha,t) = URR(\alpha,t) = ULR(\alpha,t) = \emptyset\\ 
		i + (2 \max \{\card{ULL(\alpha,t)},\card{URR(\alpha,t)}\}+ 1)
		+ (\card{X(\gamma,t)}+c(\gamma,t)),  \text{ otherwise.}
		\end{cases}
	\end{equation}

	Let us first consider the case when $ULL(\alpha,t),URR(\alpha,t)$ 
	and $ULR(\alpha,t)$ are empty. By Proposition~\ref{propPhaseOne}, we know that,
	
	\begin{equation} \label{eqnCrossedAlpha}
		\card{X(\alpha,t)} = \card{X(s,t)} -(i-1).
	\end{equation}
When $\card{ULR(\alpha,t)} = 0$, there can be at most one crossed left value in the
right half of $\alpha$ and hence $\card{X(\alpha,t)}$ is at most $1$.  From
Proposition~\ref{propPhaseTwo}, we know that the number of moves in Phase Two
of this case is at most $2$, $\card{X(\gamma,t)}\leq \card{X(\alpha,t)}+ 2$ and
$\chi(\gamma,t)=1$.  In fact, as discussed in the proof of Proposition~\ref{propPhaseTwo}, this alternating cycle is the only cycle in
$\gamma$. Therefore,
$c(\gamma,t)=1$.  Using these observations in Equation~\ref{TheoremThreeProof},
we will get
\begin{equation} \label{eqnMainResultOne}
		\orient{d} \leq \card{X(s,t)}+6.
\end{equation}

Next we consider the case when either $ULL(\alpha,t)$ or $URR(\alpha,t)$ 
or $ULR(\alpha,t)$ is non-empty. By Proposition~\ref{propPhaseTwo}, we have
\begin{equation}\label{eqnCrossGamma}
	\card{X(\gamma,t)} \leq \card{X(\alpha,t)} + 2 \max \{\card{ULL(\alpha,t)},\card{URR(\alpha,t)} \}.
\end{equation}
Since $\card{X(\alpha,t)} \leq \card{X(s,t)}-(i-1)$ (Proposition~\ref{propPhaseOne}), 
we can rewrite the Equation~(\ref{eqnCrossGamma}) as
\begin{equation}\label{eqnXGamma}
	\card{X(\gamma,t)} \leq \card{X(s,t)} -(i-1)+ 2\max \{\card{ULL(\alpha,t)}, \card{URR(\alpha,t)} \}.
\end{equation}
Since $ULL(\gamma,t)$ and $URR(\gamma,t)$ are empty, the only non-singleton 
cycles in $\gamma$ are the alternating cycles of $\gamma$ and the one which 
contain $\gamma(1)$. Hence
\begin{equation} \label{eqnGammaCycle}
	c(\gamma,t) \leq \chi(\gamma,t)+1.
\end{equation}
Since $\chi(\gamma,t)\leq \chi(\alpha,t)+1$, (Proposition~\ref{propPhaseTwo}) 
and $\chi(\alpha,t) \leq \chi(s,t)$, (Proposition~\ref{propPhaseOne}) we can 
write Equation~\ref{eqnGammaCycle} as,
\begin{equation}\label{eqnGammaCycleFinal}
	c(\gamma,t) \leq \chi(s,t)+2. 
\end{equation}
Since $ULL(\alpha,t)\subseteq ULL(s,t)$, $URR(\alpha,t) \subseteq URR(s,t)$ 
and by using Equations~(\ref{eqnXGamma}) and (\ref{eqnGammaCycleFinal}), 
we can rewrite Equation~(\ref{TheoremThreeProof}) as
\begin{equation}\label{eqnMainResultTwo}
	\orient{d} \leq   4 \max \{\card{ULL(s,t)},\card{URR(s,t)} \}+ \card{X(s,t)}+\chi(s,t)+4.
\end{equation}
Combining the results in Equation~\ref{eqnMainResultOne} and Equation~\ref{eqnMainResultTwo}, 
we can conclude that the number of steps taken by Algorithm~\ref{RouteEven} 
to send a packet from $s$ to $t$ is at most
\[
\card{X(s,t)}+\max\{6, y\}
\]

steps, where 
\[
	y=4 \max \{\card{ULL(s,t)}, \card{URR(s,t)} \} +
	 \chi(s,t) + 4.
\]
Hence Theorem~\ref{thmMain}.

\end{proof}
\begin{corollary}\label{corOne}
Let $d$ be the distance between any two nodes $s$ and $t$ in the
unoriented star graph $\SN$. Then $\orient{d}$, the distance between $s$ and $t$
in the oriented star graph $\orient{\SN}$ oriented using scheme in
\cite{fujita2013oriented} is upper bounded as,
\begin{equation}\label{eqnCorOne}
	\orient{d}  \leq 4d+4 
\end{equation}
\end{corollary}
\begin{proof}
First we consider the case when $y \leq 6$ in Theorem~\ref{thmMain}. 
In this case, $\orient{d} \leq \card{X(s,t)} + 6$. Since $d \geq \card{X(s,t)}$ 
(Equation~\ref{eqnUndirectedDistance}), $\orient{d} \leq d + 6$ which is at most $4d + 4$ for all $d \geq 1$.
Now, let us consider the case when $y > 6$ in Theorem~\ref{thmMain}. 
Hence
\begin{equation}\label{eqnCorOne4}
	\orient{d} \leq \card{X(s,t)} + 4 \max \{\card{ULL(s,t)}, \card{URR(s,t)}\} + \chi(s,t)+4.
\end{equation}
By Equation~\ref{eqnUndirectedDistance}, we know that
\begin{equation}\label{eqnCorOne1d}
		d = \card{ULL(s,t)} + \card{URR(s,t)} +\card{X(s,t)} + c(s,t).
\end{equation}
Since,
\begin{equation}
	\chi(s,t)\leq c(s,t),
\end{equation}
we can rewrite the Equation~\ref{eqnCorOne4} as,
\begin{align}
	\orient{d} &\leq  4 \max \{\card{ULL(s,t)}, \card{URR(s,t)}\} + \card{X(s,t)}+c(s,t)+4 \nonumber \\
	& \leq 4(\card{ULL(s,t)}+\card{URR(s,t)})+\card{X(s,t)}+c(s,t)+4 \nonumber \\
	& \leq 4d+4.
\end{align}
\end{proof}

\begin{corollary}\label{corTwo}
	The diameter of the oriented star graph $\orient{\SN}$,
\begin{equation}\label{eqnDiameter}
	\diam(\orient{\SN}) \leq
	\begin{cases}
		2n+2, \, \text{when } n \text{ is odd},\\
		2n+4, \, \text{otherwise.}
	\end{cases}
\end{equation}
\end{corollary}

\begin{proof}
Let $s$ and $t$ be two vertices in $\USN$ such that $\orient{d}(s,t)$ is
maximised. That is, $\diam(G) = \orient{d}(s,t)$. By Theorem~\ref{thmMain}, the number of steps that Algorithm~\ref{RouteEven} takes to send a packet from $s$ to $t$ is at most $\card{X(s,t)}+\max\{6, y\}$, where 
\[ y=4 \max\{\card{ULL(s,t)}, \card{URR(s,t)} \} + \chi(s,t) + 4.  \]
 Since $\card{X(s,t)} \leq n$, we have $\card{X(s,t)} + 5 \leq 2n + 2$ for all $n \geq 3$. So we can assume that $y \geq 6$ and therefore,
\begin{equation}
\label{eqnMainTheorem}
\begin{split}
\diam(G) \leq
	\card{X(s,t)} + 
	4 \max \{\card{ULL(s,t)}, \card{URR(s,t)} \} + \chi(s,t) + 4.
	\end{split}
\end{equation}
Suppose, $\card{ULL(s,t)} \geq \card{URR(s,t)}$.  We know that $X(s,t) =
ULR(s,t) \cup URL(s,t)$, and $\chi(s,t)$ is at most $\card{ULR(s,t)}$. Hence Eqn.~\ref{eqnMainTheorem} becomes
\begin{equation}
\label{eqnMainTheoremULLbig}
	\diam(G) \leq
	4 \card{ULL(s,t)} + 2\card{ULR(s,t)} + \card{URL(s,t)} + 4.
\end{equation}
Since, $\card{URL(s,t)}$ right values of $t$ are not in the right half of 
$s$, at least $(\card{URL(s,t)} - 1)$ of them are in the left half of $s$. 
That is $\card{ULR(s,t)} \geq \card{URL(s.t)} - 1$. We can improve this 
bound if $t(1)$ is not in the left half of $s$.  If $t(1)$ is in the first 
position of $s$, then $\card{URL(s,t)} = \card{ULR(s,t)}$.  If $t(1)$ is in 
the right half of $s$, then $(\card{URL(s,t)} + 1)$ right values of $t$ are 
not in the right half of $s$, and hence at least $\card{URL(s,t)}$ of them 
are in the left half of $s$.
Therefore,
\begin{equation}
\label{eqnMainTheoremULLbig2}
\diam(G) \leq
\begin{cases}
	4 \card{ULL(s,t)} + 3\card{ULR(s,t)} + 5, \text{ if } t(1) \in L(s) \\
	4 \card{ULL(s,t)} + 3\card{ULR(s,t)} + 4, \text{otherwise.}
\end{cases}
\end{equation}
Since $\card{ULL(s,t)} + \card{ULR(s,t)}$ is at most $\card{L(s)} - 1$ in the former case above and $\card{L(s)}$ in the latter case above, we get
\[
\diam(G) \leq 4 \card{L(s)} + 4.
\] 
Since $\card{L(s)}$ is $(n-1)/2$ when $n$ is odd and $n/2$ when 
$n$ is even, we satisfy Equation~\ref{eqnDiameter}. The case when 
$\card{URR(s,t)} > \card{ULL(s,t)}$ is similar and hence omitted.
\end{proof}
Day and Tripathi have numerically computed the diameter of $\USN$ for 
$n$ in the range $3$ to $9$ \cite{day1993unidirectional}. Our bounds on 
the diameter of $\USN$ agrees with their computation when $n \leq 8$. 
For $n=9$ our upper bound is $20$ while their computation reports $24$.

 \bibliographystyle{splncs04}
 \bibliography{star}
\end{document}